\documentclass[conference]{IEEEtran}
\IEEEoverridecommandlockouts

\usepackage{cite}
\usepackage{amsmath,amssymb,amsfonts}
\usepackage{amsthm}
\newtheorem{lemma}{Lemma}
\usepackage{algorithmicx}
\usepackage{algpseudocode}
\usepackage{graphicx}
\usepackage{textcomp}
\usepackage{xcolor}
\usepackage{float}
\usepackage{dblfloatfix}
\usepackage[ruled,lined,linesnumbered]{algorithm2e}
\usepackage{enumitem}
\usepackage{soul}

\def\BibTeX{{\rm B\kern-.05em{\sc i\kern-.025em b}\kern-.08em
    T\kern-.1667em\lower.7ex\hbox{E}\kern-.125emX}}

\begin{document}

\title{Precision Coexistence: Angular-Spectral Region Constrained Waveform Design for MIMO-ISAC}

\author{
\IEEEauthorblockN{
Syed Ali Haider Shirazi$^{*}$, Taejoon Kim$^{*}$, David J. Love$^{\dagger}$, and Sergiy A. Vorobyov$^{\ddagger}$
}
\IEEEauthorblockA{
$^{*}$School of Electrical, Computer and Energy Engineering, Arizona State University, USA\\
$^{\dagger}$Elmore Family School of Electrical and Computer Engineering, Purdue University, USA\\
$^{\ddagger}$Department of Information and Communications Engineering, Aalto University, Finland\\
Email: \{sshiraz6, taejoonkim\}@asu.edu, djlove@purdue.edu, sergiy.vorobyov@aalto.fi
}
\thanks{This work is in part supported by the National Science Foundation (NSF) under grants CNS2451268, CNS2514415, CNS222557, and ITE2515378, and the Office of Naval Research (ONR) under Grant N000142112472.}%
}

\maketitle

\begin{abstract}
In this paper, we develop a dual-function waveform design framework for multiple-input multiple-output integrated sensing and communication (MIMO-ISAC) systems operating in congested electromagnetic environments. In such environments, the dual-function waveform must coexist with legacy systems, such as localized positioning sensors, maritime radar networks, and satellite communication links, that occupy specific frequency bands and angular sectors. To enable precision coexistence, we incorporate localized angular-spectral region constraints that suppress beamformed energy only within the angle-frequency regions occupied by each legacy system, while preserving the spatial degrees of freedom (DoFs) available outside those regions. The resulting waveform design problem is formulated as a nonconvex optimization problem and solved using a low-complexity majorization--minimization alternating direction method of multipliers (MM--ADMM) algorithm. Simulation results demonstrate that the proposed framework provides precise legacy-system protection while maintaining robust radar sensing and multi-user communication performance.
\end{abstract}
\begin{IEEEkeywords}
Integrated sensing and communication (ISAC), 
waveform design, precision coexistence.
\end{IEEEkeywords}

{\renewcommand{\thefootnote}{}
\footnotetext{Accepted for presentation at IEEE Conference on Integrated Sensing and Communications 2026.}
}

\section{Introduction}
Integrated sensing and communication (ISAC) has emerged as a key
enabling technology for next-generation wireless systems, where a single
multiple-input multiple-output (MIMO) platform simultaneously performs radar
sensing and multi-user communication~\cite{ISAC_1}. A fundamental challenge
in practical ISAC deployment arises in congested electromagnetic environments,
where the dual-function waveform must coexist with legacy systems, such as
localized positioning sensors and radar altimeter receivers, occupying specific
frequency bands and angular sectors~\cite{aubry1,LEGACY1}. At the same time,
the waveform must maintain adequate radar signal-to-clutter-plus-noise ratio
(SCNR) and communication quality-of-service (QoS) for all served users.
Since legacy systems are localized in both frequency and space, effective
interference mitigation requires joint angular-spectral control over the
transmitted energy.

\begin{figure}[t]
  \centering
  \includegraphics[width=1.0\linewidth,height=0.14\textheight]{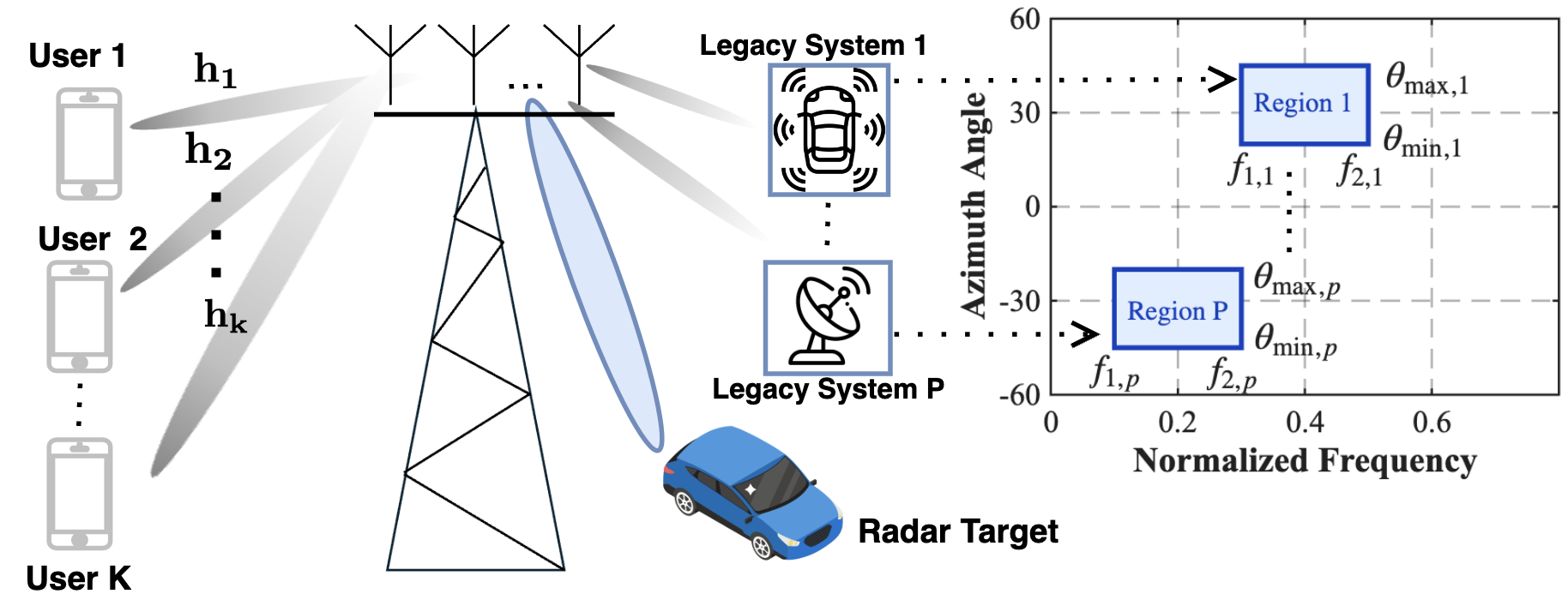}
  \caption{MIMO-ISAC system serving multiple downlink users, sensing a radar
  target, and protecting legacy systems.}
  \label{fig:spatial_spectral_regions}
\end{figure}

Conventional radar waveform design approaches typically enforce spectral
coexistence via frequency-domain notching constraints that suppress energy across
entire protected bands regardless of the angular direction~\cite{ADMMWave,aubry1}.
While effective for spectrum protection, such methods
underutilize the spatial degrees of freedom (DoFs) available in MIMO
architectures. With the advent of MIMO radar, joint angular-spectral waveform
design has been explored to shape the sensing only beam energy over selected angle-frequency
regions~\cite{mimo_radar}. Whereas from the communication perspective, spatial nulling for legacy users protection  has been studied through precoder design methods~\cite{LEGACY1,sameer2}, which are blind to the spectral variations of the environment.

More recently, angular--spectral coexistence has been considered in joint
radar--communication systems, where dedicated radar and communication
transmitters are jointly optimized~\cite{jointRadar_comm}. However, these
coexistence architectures rely on separate hardware and waveforms for
sensing and communication, increasing implementation complexity and
hardware cost while potentially introducing mutual interference. This motivates a unified ISAC architecture in which a single
waveform generated by a common base station simultaneously supports sensing,
communication, and legacy-system coexistence. In particular, a MIMO-Orthogonal frequency division multiplexing (MIMO-OFDM) waveform provides a practical basis for such integration due to its compatibility with existing 5G wireless architectures and standards.

Motivated by this need, this paper develops a dual-function MIMO-OFDM ISAC waveform design framework with joint angular-spectral constraints that selectively suppress beamformed energy only within the angle-frequency regions occupied by each
legacy system, thereby preserving valuable MIMO spatial DoFs outside those
regions. Within this framework, SCNR is adopted as the primary sensing
objective to promote robust target detection~\cite{ISAC_2,1iu2}. However,
the coupling between the transmit waveform and the waveform-dependent clutter
covariance, together with the communication QoS and coexistence constraints,
renders the resulting optimization problem highly nonconvex. To address this problem, we develop a low-complexity algorithm that integrates
the majorization--minimization (MM) framework~\cite{MM,lee_mm,lee_mm2} with
the alternating direction method of multipliers
(ADMM)~\cite{boydADMM,lee2025_clutter}. Simulation
results demonstrate precise localized protection across multiple
angular-spectral regions with negligible degradation in radar SCNR and
communication QoS relative to the unconstrained benchmark~\cite{ISAC_2}.

\textit{Notation:} $(\cdot)^T$, $(\cdot)^H$, $(\cdot)^*$, $\|\cdot\|_2$, and $\Re\{\cdot\}$ and $\Im\{\cdot\}$ denote, respectively, the
transpose, conjugate transpose,  complex conjugate, the Euclidean ($\ell_2$) norm, and the real and imaginary parts of a complex scalar or vector. 
Moreover, $\mathbf{I}_n$, $\mathbf{0}_n$, $\otimes$, $\operatorname{rank}\{\cdot\}$, and $\operatorname{blkdiag}\{\cdot\}$ and  $\operatorname{diag}(\cdot)$ denote, respectively, the $n\times n$ identity matrix, $n\times 1$ zero vector, the Kronecker product, the rank of a matrix, and block-diagonal matrix and diagonal matrix formed from their arguments. 
For a vector $\mathbf{u} \in \mathbb{C}^{M \times 1}$, 
$\{\mathbf{u}\}_{[i,j]}$ denotes the sub-vector comprising 
entries $i$ through $j$. For a matrix $\mathbf{U} \in 
\mathbb{C}^{M \times M}$, $\{\mathbf{U}\}_{[i,j],\,[p,q]}$ 
denotes the submatrix formed by rows $i$ through $j$ and 
columns $p$ through $q$.

\section{System Model}
We consider a monostatic multiuser OFDM MIMO-ISAC system depicted 
in Fig.~\ref{fig:spatial_spectral_regions}, where the base station~(BS) 
is equipped with collocated $N_T$ transmit and $N_R$ receive antennas. 
The system operates over $N$ orthogonal subcarriers with subcarrier 
spacing $\delta_f$, indexed by $n \in \{0,\ldots,N{-}1\}$, and transmits 
$L$ consecutive OFDM symbols indexed by 
$\ell \in \{0,\ldots,L{-}1\}$, during one coherent processing 
interval~(CPI). The OFDM symbol duration, $T_{\mathrm{sym}}$, includes both the useful
symbol duration $1/\delta_f$ and the cyclic prefix duration
$T_{\mathrm{CP}}$.
The spatial transmit vector
on subcarrier~$n$ during symbol~$\ell$ is
$\boldsymbol{\omega}_{n,\ell}\in\mathbb{C}^{N_T\times 1}$. After cyclic prefix (CP) removal and $N$-point discrete Fourier
transform~(DFT) processing at the radar receiver, the complex
baseband signal received across the $N_R$ receive antennas on
subcarrier~$n$ during symbol~$\ell$ is
\begin{equation}
  \label{eq:rx_per_sc}
  \boldsymbol{y}_{n,\ell}
  = \mathbf{H}^{\mathrm{t}}_{n,\ell}\,\boldsymbol{\omega}_{n,\ell}
  + \mathbf{H}^{\mathrm{c}}_{n,\ell}\,\boldsymbol{\omega}_{n,\ell}
  + \mathbf{z}_{n,\ell},
\end{equation}
where $\boldsymbol{y}_{n,\ell}\in\mathbb{C}^{N_R\times 1}$, while 
$\mathbf{H}^{\mathrm{t}}_{n,\ell}\in\mathbb{C}^{N_R\times N_T}$
and $\mathbf{H}^{\mathrm{c}}_{n,\ell}\in\mathbb{C}^{N_R\times N_T}$
are the target and clutter channel matrices, respectively. 
The vector $\mathbf{z}_{n,\ell}\in\mathbb{C}^{N_R\times 1}$ is independent 
and identically distributed~(i.i.d.) complex circular zero-mean Gaussian noise vector with variance
$\sigma_r^2$, i.e., $\mathbf{z}_{n,\ell}\sim\mathcal{CN}(\mathbf{0}_{N_R},
\sigma_r^2\mathbf{I}_{N_R})$.  

The target and clutter channels in \eqref{eq:rx_per_sc} share an identical
structure due to superposed environmental returns: 
\begin{align}
  \mathbf{H}^{\mathrm{t}}_{n,\ell}
  &=\kappa_0\,
     e^{\,j2\pi(f_{D,0}\ell T_{\mathrm{sym}} - n\delta_f\tau_0)}\,
    \boldsymbol{\psi}_{\mathrm{r}}(\theta_0)\boldsymbol{\psi}_{\mathrm{t}}^H(\theta_0),
  \label{eq:Ht}\\[3pt]
  \mathbf{H}^{\mathrm{c}}_{n,\ell}
  &= \sum_{c=1}^{C}
     \kappa_c\,
     e^{\,j2\pi(f_{D,c}\ell T_{\mathrm{sym}} - n\delta_f\tau_c)}\,
    \boldsymbol{\psi}_{\mathrm{r}}(\theta_c)\boldsymbol{\psi}_{\mathrm{t}}^H(\theta_c),
  \label{eq:Hc}
\end{align}
where $\kappa_0\in\mathbb{C}$ captures the combined effects of radar 
cross section~(RCS) and path loss. The target is characterized by its azimuth angle 
$\theta_0$,  
Doppler frequency $f_{D,0}$, induced by its radial 
velocity $v_0$ and round-trip delay $\tau_0 = 2R_0/\zeta$, where $R_0$ is the target range 
and $\zeta$ is the speed of light. 
The clutter scatterer $c$ is parameterized identically 
by $\kappa_c$, $f_{D,c}$, $\tau_c$, and $\theta_c$. For a 
half-wavelength-spaced uniform linear array~(ULA), the transmit and receive steering vectors are $\boldsymbol{\psi}_{\mathrm{t}}(\theta) \!\!=\!\! [1, 
e^{-j\pi\sin\theta}, \ldots, e^{-j(N_T-1)\pi\sin\theta}]^T 
\!\!\!\in\!\mathbb{C}^{N_T\times 1}$ and $\boldsymbol{\psi}_{\mathrm{r}}(\theta) \!\!=\!\! [1, 
e^{-j\pi\sin\theta}, \ldots, e^{-j(N_R-1)\pi\sin\theta}]^T 
\!\!\in\! \mathbb{C}^{N_R\times 1}$, respectively.


To enable coherent processing across the $L$ OFDM symbols, the per-symbol received vectors on subcarrier $n$ in \eqref{eq:rx_per_sc} are organized into a unified space-time block vector $\boldsymbol{y}_n =
[\boldsymbol{y}_{n,0}^T,\ldots,\boldsymbol{y}_{n,L{-}1}^T]^T
\in\mathbb{C}^{N_RL\times 1}$. 
The Doppler-induced phase progression
across $L$ symbols is
captured by $\mathbf{d}(f_D) = [1, e^{j2\pi f_D T_{\mathrm{sym}}},
\ldots, e^{j2\pi f_D(L-1)T_{\mathrm{sym}}}]^T
\in\mathbb{C}^{L\times 1}$, and spatial-temporal steering vector $\mathbf{v}(\theta,f_D)
  = \mathbf{d}(f_D)\otimes\boldsymbol{\psi}_r(\theta)_\otimes\boldsymbol{\psi}_{\mathrm{t}}(\theta)
    \in\mathbb{C}^{N_TN_RL\times 1}$.
Then, the stacked received signal on subcarrier $n$ admits the compact form
\begin{align}
  \label{eq:yn_compact}
  \boldsymbol{y}_n
  &=\kappa_0\,e^{-j2\pi n\delta_f\tau_0}\,
    \boldsymbol{W}_{n}\mathbf{v}(\theta_0,f_{D,0})
  + \mathbf{z}_n \notag\\
  &\quad \hspace{1cm}+ \sum_{c=1}^{C}
    \kappa_c\,e^{-j2\pi n\delta_f\tau_c}\,
    \boldsymbol{W}_{n}\mathbf{v}(\theta_c,f_{D,c}),
\end{align}
where the first term on the right hand side (r.h.s.) of \eqref{eq:yn_compact} represents the target return, where
$\boldsymbol{W}_{n} \!=\! \operatorname{blkdiag}\{
\mathbf{I}_{N_R}\otimes\boldsymbol{\omega}_{n,0}^T, \ldots,
\mathbf{I}_{N_R}\otimes\boldsymbol{\omega}_{n,L{-}1}^T
\} \!\in\! \mathbb{C}^{N_RL\times N_TN_RL}$ is the transmit matrix
on subcarrier $n$, $e^{-j2\pi n\delta_f\tau_0}$ encodes the range-dependent 
phase rotation across subcarriers, and $\boldsymbol{W}_{n}\mathbf{v}(\theta_0,
f_{D,0})$ is the space-time response of the transmitted waveform
to the target at angle $\theta_0$ and Doppler $f_{D,0}$. 
The last term on the r.h.s. of 
\eqref{eq:yn_compact} accumulates the returns from $C$ clutter scatterers,
each modulating the transmitted waveform through its own
space-time vector $\mathbf{v}(\theta_c,f_{D,c})$ and
subcarrier phase rotation $e^{-j2\pi n\delta_f\tau_c}$. 
The vector
$\mathbf{z}_n = [\mathbf{z}_{n,0}^T,\ldots,\mathbf{z}_{n,L{-}1}^T]^T
\in\mathbb{C}^{N_RL\times 1}$ is the stacked noise.

As depicted in Fig. \ref{fig:spatial_spectral_regions}, the BS also serves $K$ single-antenna communication users over the identical $N$ subcarriers and $L$ OFDM symbols allocated for radar sensing. The signal received at user $k$ on subcarrier $n$ of symbol $\ell$
is
\begin{equation}
  \label{eq:rx_user}
  y_{n,\ell,k}
  = \mathbf{h}_{n,k}^{H}\,\boldsymbol{\omega}_{n,\ell}
  + \nu_{n,\ell,k},
\end{equation}
where $\mathbf{h}_{n,k}\in\mathbb{C}^{N_T\times 1}$ is the
downlink channel vector from the BS to user $k$ on subcarrier $n$, assumed to be known at the transmitter, and $\nu_{n,\ell,k}\sim\mathcal{CN}(0,\sigma_c^2)$ is the  additive noise at each communication user. 
The intended symbol $s_{n,\ell,k}$ for user $k$ on subcarrier $n$ of symbol $\ell$ is drawn
independently from an $\Omega$-phase shift keying (PSK) constellation
$\mathcal{S} = \{ e^{j\pi(2i-1)/\Omega}, i = 1,\ldots,\Omega \}$,
where $\Phi = \pi/\Omega$ is half the angular width of each
decision region.

Symbol-level precoding~(SLP)~\cite{ISAC_2} is well-suited here 
because it directly optimizes instantaneous multi-carrier waveforms rather than static linear precoding.
By exploiting constructive interference~(CI), SLP ensures the rotated received signal $\mathbf{h}_{n,k}^H\boldsymbol{\omega}_{n,\ell}s_{n,\ell,k}^*$ lies strictly within the valid PSK decision boundaries with a predetermined safety margin $\gamma_k$ of each user $k$.
The CI margin for user $k$ on subcarrier $n$ at symbol $\ell$ is defined as \cite{ISAC_2}
\begin{align}
  \label{eq:ci_margin}
  \delta_{n,\ell,k}
  &= \Re\!\left\{\mathbf{h}_{n,k}^H\boldsymbol{\omega}_{n,\ell}
    s_{n,\ell,k}^*\right\}\sin(\pi/\Omega) \notag\\
  &\quad \hspace{1cm} - \left|\Im\!\left\{\mathbf{h}_{n,k}^H\boldsymbol{\omega}_{n,\ell}
    s_{n,\ell,k}^*\right\}\right|\cos(\pi/\Omega).
\end{align}
The QoS requirement $\delta_{n,\ell,k}\!\geq \!\gamma_{k}$
ensures the received signal falls within the correct PSK decision
region across all subcarriers $n$, symbols $\ell$, and users $k$. 
The margin is defined as
$\gamma_k = \sqrt{\sigma_c^2\,\Gamma}\sin(\pi/\Omega)$,
where $\Gamma \geq 0$ is the
communication QoS threshold and $\sigma_c^2$ the receiver
noise variance.

\section{Angular--Spectral Region Constraints}
Coexistence between the ISAC system and incumbent legacy systems requires limiting ISAC radiated energy within the joint angle--frequency regions of $P$ legacy systems, where each legacy system $p \in \{1,\ldots,P\}$ is defined by an angular sector $[\theta_{\min,p},\theta_{\max,p}]$ and frequency interval $[f_{1,p},f_{2,p}]$, as illustrated in Fig.~\ref{fig:spatial_spectral_regions}.
We constrain the aggregate energy radiated into these continuous protected regions. 

The angular component of the coexistence constraint for legacy system~$p$
is captured by the positive semidefinite spatial correlation matrix
$\mathbf{B}_p\in\mathbb{C}^{N_T\times N_T}$, defined as
$\mathbf{B}_p=\frac{1}{N_T}\int_{\theta_{\min,p}}^{\theta_{\max,p}}
\boldsymbol{\psi}_{\mathrm{t}}(\theta)
\boldsymbol{\psi}_{\mathrm{t}}^H(\theta)\,d\theta$, where
$\boldsymbol{\psi}_{\mathrm{t}}(\theta)$ is the transmit steering vector
defined in Section~II. The matrix $\mathbf{B}_p$ captures the angular exposure across the protected sector.

Similarly, the spectral component is
captured by $\mathbf{S}_{I,p}\in\mathbb{C}^{N\times N}$, which integrates the transmitted spectral energy over the continuous band $[f_{1,p}, f_{2,p}]$ and its  
$m$th row and $n$th column entry is given by 
$[\mathbf{S}_{I,p}]_{m,n} = \int_{f_{1,p}}^{f_{2,p}}
e^{j2\pi f(m-n)}\,df$, for $m,n\in\{0,1,\ldots,N{-}1\}$. The spectral energy is evaluated over the useful OFDM symbol duration $1/\delta_f$. Since $\boldsymbol{\omega}_{n,\ell}$ is in the frequency domain,  we account for the unitary IFFT mapping through $\mathbf{R}_{I,p}
  = \mathbf{F}_N\mathbf{S}_{I,p}\mathbf{F}_N^H$. Here $\mathbf{F}_N\in\mathbb{C}^{N\times N}$ is the unitary DFT matrix with entries $[\mathbf{F}_N]_{m,n}=N^{-1/2}e^{-j2\pi mn/N}$. Thus, $\mathbf{F}_N^H$ maps subcarrier coefficients to time-domain samples.
Unlike conventional subcarrier-selective masking techniques~\cite{LM}, ${\mathbf{R}}_{I,p}\succeq\mathbf{0}$ integrates over the continuous band $[f_{1,p},f_{2,p}]$ to capture both in-band energy and the spectral sidelobe leakage caused by the finite temporal duration of the OFDM waveform.

We formulate a joint angular--spectral energy matrix, unifying the spatial and spectral domains into a single matrix via the Kronecker product $\mathbf{Q}_p
  = \mathbf{R}_{I,p}\otimes\mathbf{B}_p
  \in\mathbb{C}^{N_TN\times N_TN},\mathbf{Q}_p\succeq\mathbf{0}.$
By stacking the per-subcarrier transmit vectors across all $N$ subcarriers into a space--frequency vector $\boldsymbol{\omega}_{\ell}\triangleq
[\boldsymbol{\omega}_{0,\ell}^T,\ldots,\boldsymbol{\omega}_{N-1,\ell}^T]^T
\in\mathbb{C}^{N_TN\times 1}$ for symbol $\ell$, the quadratic form $\boldsymbol{\omega}_{\ell}^H\mathbf{Q}_p\boldsymbol{\omega}_{\ell}$ yields the total beamformed energy of the waveform within the protected angle--frequency region $[\theta_{\min,p},\theta_{\max,p}]\times[f_{1,p},f_{2,p}]$ of legacy system $p$.
Consequently, multi-domain coexistence is enforced by limiting this leakage:
\begin{equation}
  \label{eq:coexistence}
  \boldsymbol{\omega}_{\ell}^H\mathbf{Q}_p
  \boldsymbol{\omega}_{\ell}
  \leq \eta_p,
  \qquad \forall \ell,\ \forall p,
\end{equation}
where $\eta_p>0$ is the maximum allowable interference-energy threshold for system $p$. This unified formulation generalizes across diverse array geometries and protected frequency intervals.

\section{Problem Formulation}
To process the radar backscatter signal, a linear filter $\mathbf{g}_n \in \mathbb{C}^{N_RL\times 1}$ is applied
to the stacked received vector $\boldsymbol{y}_n$
in~\eqref{eq:yn_compact}. The resulting radar output SCNR on subcarrier $n$ is
\begin{equation}
  \label{eq:sinr_n}
  \mathrm{SCNR}_n(\boldsymbol{W}_{n}, \mathbf{g}_n)
  =
  \frac{
    \sigma_0^2\,
    \left|\mathbf{g}_n^H\boldsymbol{W}_{n}\mathbf{v}(\theta_0,f_{D,0})\right|^2
  }{
    \mathbf{g}_n^H\mathbf{R}_n(\boldsymbol{W}_{n})\mathbf{g}_n
  },
\end{equation}
where $\mathbf{R}_n$ is the  clutter-plus-noise covariance matrix:  
\begin{equation}
  \label{eq:Rn}
  \!\!\mathbf{R}_n
  \!=\!\!
  \sum_{c=1}^{C}\!|\kappa_c|^2 
  \boldsymbol{W}_{n}\mathbf{v}(\theta_c,f_{D,c})
  \mathbf{v}(\theta_c,f_{D,c})^H\boldsymbol{W}_{n}^{H}
  \!+\! \sigma_r^2\mathbf{I}_{N_RL}. 
\end{equation}
For any given $\boldsymbol{W}_{n}$, the minimum variance distortionless
response~(MVDR) filter $\mathbf{g}_n$ that maximizes~\eqref{eq:sinr_n} is expressed as $\mathbf{g}_n = \mathbf{R}_n^{-1}\boldsymbol{W}_{n} \mathbf{v}(\theta_0, f_{D,0})$ yielding the maximized SCNR
\begin{equation}
  \label{eq:sinr_mvdr}
\!\!\!\!  \mathrm{SCNR}_n\!(\boldsymbol{W}_{n})
  \!=\!
  \sigma_0^2\,
  \mathbf{v}(\theta_0,f_{D,0})^H
  \boldsymbol{W}_{n}^{H}
  \mathbf{R}_n^{-1}
  \boldsymbol{W}_{n}
  \mathbf{v}(\theta_0,f_{D,0}).
\end{equation}

The joint waveform design for multi-domain coexistence problem is then formulated as the following optimization problem: 
\begin{equation}
  \label{eq:problem}
  \max_{\{\boldsymbol{\omega}_{n,\ell}\} \in \mathcal{W}}\;\;
  \sum_{n=0}^{N-1}\mathrm{SCNR}_n(\boldsymbol{W}_{n}),
\end{equation}
where the feasible set $\mathcal{W}$ is defined as 
\begin{equation}
\!\!\!\!
  \label{eq:feasible}
  \mathcal{W}
  \!\!=\!
  \left\{ \!\!
    \{\boldsymbol{\omega}_{n,\ell}\}
    \; \!\!\middle|\; \!\!\!\!
    \begin{array}{l}
      \delta_{n,\ell,k}\geq\gamma_{k},\;
      (\mathbf{\omega}_\ell)^H\mathbf{Q}_p\mathbf{\omega}_\ell
      \leq \eta_p,\\[3pt]
      |{\omega}_{r,\ell,q}| \leq A_{\mathrm{papr}},\;
\!      \sum_{\ell=0}^{L-1}\|\mathbf{\omega}_\ell\|_2^2 \!=\! P_T
    \end{array}
  \!\!\!\!\right\},
\end{equation}
with ${\omega}_{r,\ell,\tau}$ being the $\tau$-th time-domain sample  at transmit antenna element $r$ of OFDM symbol $\ell$, obtained
via an $N$-point inverse DFT of $\{\boldsymbol{\omega}_{n,\ell}\}_{n=0}^{N-1}$. 
In \eqref{eq:feasible}, the first constraint enforces the CI-QoS
requirements for all $K$ users, while the second constraint  implements our proposed joint angular-spectral coexistence condition via $\mathbf{Q}_p$ in \eqref{eq:coexistence} to bound leakage energy below the threshold $\eta_p$.  The third is the global
peak-to-average power ratio~(PAPR) constraint, which limits the
instantaneous time-domain amplitude at each antenna below
$A_{\mathrm{papr}}=\sqrt{(1+\varepsilon)P_T/(LNN_T)}$
with tolerance $\varepsilon\geq 0$, mitigating non-linear distortions in the radar RF chains.
The last constraint normalizes the total transmit power to $P_T$.
To tackle the non-convex problem in \eqref{eq:problem}, we develop a low-complexity MM-ADMM algorithm in the next section. 

\section{Proposed Algorithm}
A challenge in solving~\eqref{eq:problem} is the waveform-dependent clutter covariance $\mathbf{R}_n$ in~\eqref{eq:Rn}.
To address this, we propose a double-loop iterative algorithm (Algorithm~\ref{alg:mm_admm}) based on a MM-ADMM framework. Below, we establish preliminaries which enable its computationally efficient implementation. 

\subsection{Preliminaries}
\label{sec:preliminary}

For a given waveform $\boldsymbol{W}_{n}$, the MVDR receive filter is $\mathbf{u}_n \triangleq
\mathbf{R}_n^{-1}\boldsymbol{W}_{n}
\mathbf{v}(\theta_0,f_{D,0}) \in \mathbb{C}^{N_RL\times 1}$, and the invariant spatial clutter kernel is
$\boldsymbol{\Lambda}_n \triangleq \sum_{c=1}^{C}
|\kappa_c|^2\mathbf{v}(\theta_c,f_{D,c})
\mathbf{v}(\theta_c,f_{D,c})^H
\in \mathbb{C}^{N_TN_RL\times N_TN_RL}.$ 
We define $\boldsymbol{\omega}_{n} \triangleq
[\boldsymbol{\omega}_{n,0}^T, \ldots, \boldsymbol{\omega}_{n,L-1}^T]^T \in
\mathbb{C}^{N_TL \times 1}$ as the per-subcarrier waveform vector stacked across all $L$ symbols. 
The following lemma reformulates the waveform dependence inside the clutter covariance matrix $\mathbf{R}_n$ in \eqref{eq:Rn} into an explicit quadratic form of $\boldsymbol{\omega}_{n}$.

\begin{lemma}
\label{lem:decouple}
Let $\boldsymbol{\Psi}_n \triangleq \operatorname{blkdiag}
\{
\mathbf{u}_{n,0} \otimes \mathbf{I}_{N_T},
\ldots,
\mathbf{u}_{n,L-1} \otimes \mathbf{I}_{N_T}
\}
\in \mathbb{C}^{N_TN_RL \times N_TL}$, where
$\mathbf{u}_{n,\ell} \triangleq {\mathbf{u}_n}_{[\ell N_R + 1,(\ell+1)N_R]} \in \mathbb{C}^{N_R \times 1}$ is the $\ell$-th block of $N_R$ consecutive entries of $\mathbf{u}_n
=
\mathbf{R}_n^{-1}
\boldsymbol{W}_{n}
\mathbf{v}(\theta_0,f_{D,0})
\in \mathbb{C}^{N_RL \times 1}$.
The clutter quadratic form in~\eqref{eq:Rn} can be expressed as an explicit function of $\boldsymbol{\omega}_{n}$, 
\begin{equation}
\label{eq:decouple}
(\mathbf{u}_n)^H
\boldsymbol{W}_{n}
\boldsymbol{\Lambda}_n
\boldsymbol{W}_{n}^H
\mathbf{u}_n
=
\boldsymbol{\omega}_{n}^H
{\mathbf{U}}_n
\boldsymbol{\omega}_{n},
\end{equation}
where ${\mathbf{U}}_n \triangleq (\overline{\mathbf{U}}_n)^T \succeq \mathbf{0}$
and $\overline{\mathbf{U}}_n \triangleq 
\boldsymbol{\Psi}_n^H \boldsymbol{\Lambda}_n \boldsymbol{\Psi}_n
\in \mathbb{C}^{N_TL \times N_TL}$.
\end{lemma}
\begin{proof}
Using the block-diagonal structure of $\boldsymbol{W}_{n}$, the
product $\boldsymbol{W}_{n}^H\mathbf{u}_n$ can be written as
%
$\boldsymbol{W}_{n}^H\mathbf{u}_n
=
\boldsymbol{\Psi}_n\boldsymbol{\omega}_{n}^*$,
%
which follows from the inner product identity
%
$\left(\mathbf{I}_{N_R}\otimes (\boldsymbol{\omega}_{n,\ell})^T\right)^H
\mathbf{u}_{n,\ell}
=
\left(
\mathbf{u}_{n,\ell}
\otimes \mathbf{I}_{N_T}
\right)
\boldsymbol{\omega}_{n}^* $.
%
Substituting this into the quadratic form yields 
%
$\mathbf{u}_n^H
\boldsymbol{W}_{n}
\boldsymbol{\Lambda}_n
\boldsymbol{W}_{n}^H
\mathbf{u}_n
=
(\boldsymbol{\Psi}_n\boldsymbol{\omega}_{n}^*)^H
\boldsymbol{\Lambda}_n
(\boldsymbol{\Psi}_n\boldsymbol{\omega}_{n}^*)
=\boldsymbol{\omega}_{n}^T  \boldsymbol{\Psi}_n^H \boldsymbol{\Lambda}_n 
\boldsymbol{\Psi}_n \boldsymbol{\omega}_{n}^*
= \boldsymbol{\omega}_{n}^T \overline{\mathbf{U}}_n \boldsymbol{\omega}_{n}^*
=
\boldsymbol{\omega}_{n}^H
{\mathbf{U}}_n
\boldsymbol{\omega}_{n}. 
$
This completes the proof.
\end{proof}

Lemma~\ref{lem:decouple} isolates the transmit waveform from the clutter  structure, allowing us to apply the standard MM identity~\cite{MM} to construct a tractable surrogate for the radar SCNR:
\begin{equation}
\label{eq:surrogate}
\!\!\mathrm{SCNR}_n(\boldsymbol{\omega}_{n})
\ge
2\sigma_0^2
\Re
\left\{
\mathbf{q}_n^H
\boldsymbol{\omega}_{n}
\right\}
- 
\sigma_0^2
\boldsymbol{\omega}_{n}^H
{\mathbf{U}}_n
\boldsymbol{\omega}
+
c_n,
\end{equation}
where $\mathbf{q}_n \in \mathbb{C}^{N_TL\times 1}$
is a coefficient vector satisfying
$
\mathbf{u}_n^H
\boldsymbol{W}_{n}
\mathbf{v}(\theta_0,f_{D,0})
=
\mathbf{q}_n^H
\boldsymbol{\omega}_{n},
$
and $c_n$ is a constant independent of $\boldsymbol{\omega}_{n}$.
The first term on the r.h.s.\ of~\eqref{eq:surrogate} captures
the target gain at $(\theta_0,f_{D,0})$, while the second term penalizes spatial clutter leakage through the fixed receive
filter.
\begin{algorithm}[t]
\small
\caption{Proposed MM-ADMM Algorithm}
\label{alg:mm_admm}
\SetAlgoLined
\label{inputs}\KwIn{$\{\kappa_c,\theta_c,\tau_c,f_{D,c}\}$,
$\{\kappa_0,\theta_0,\tau_0,f_{D,0}\}$,
$\{\mathbf{h}_{n,k}\}$, $\{s_{n,\ell,k}\}$,
$\{\eta_p\}$, $\rho>0$, $T$, $M$}
\KwOut{$\{\mathbf{\omega}^\star_\ell\}_{\ell=0}^{L-1}$}

\textbf{Offline:} Compute $\boldsymbol{\Lambda}_n$;

\textbf{Initialize:} feasible $\{\mathbf{\omega}^{(0)}_\ell\}$\;

\For{$t=1,\ldots,T$}{\label{outerloop_t}
    \For{$n=0,\ldots,N-1$}{ 
        Compute $({\mathbf{R}}_n^{(t)})^{-1}$ via
        Sherman--Morrison recursion~\eqref{eq:SM}\; \label{Sherman_Morrison}
        Form ${\mathbf{u}}_n^{(t)}$,
        $\boldsymbol{\Psi}_n^{(t)}$,
        ${\mathbf{U}}_n^{(t)}$, and
        ${\mathbf{q}}_n^{(t)}$\; \label{MM_variables}
    }
    Stack ${\mathbf{q}}_\ell^{(t)}$ and
    ${\mathbf{U}}^{(t)}_\ell$ over $n$,
    $\forall\,\ell$\; \label{stack_waveform}

   Set $\mathbf{\omega}^{<1>}_\ell\leftarrow\mathbf{\omega}^{(t)}_\ell$,
$\mathbf{z}^{<1>}_\ell\leftarrow\mathbf{\omega}^{(t)}_\ell$,
$\boldsymbol{\lambda}^{<1>}_\ell\leftarrow\mathbf{0}_{N_TN}$,
$\forall\,\ell$\; \label{admm_initialization}

    \For{$m=1,\ldots,M$}{ \label{admm_m}
        \For{$\ell=0,\ldots,L-1$}{
            Update $\mathbf{\omega}_\ell^{<m+1>}$
            via~\eqref{eq:x_update}\; \label{admm_xupdate}
            Update $\mathbf{z}_\ell^{<m+1>}$
            via~\eqref{eq:y_update}\; \label{admm_yupdate}
            Update $\boldsymbol{\lambda}_\ell^{<m+1>}$
            via~\eqref{eq:dual_update}\; \label{admm_dualupdate}
        }
    }\label{admm_end}
    Set $\boldsymbol{\omega}_\ell^{(t+1)}\leftarrow
    \mathbf{\omega}_\ell^{<m+1>}$, $\forall\,\ell$\;
}
\textbf{Return} $\mathbf{\omega}^\star_\ell\leftarrow
\boldsymbol{\omega}_\ell^{(t+1)}$, $\forall\,\ell$\;
\end{algorithm}

%

Computing the terms on the right-hand side (RHS) of~\eqref{eq:surrogate}
requires the online inversion of the full-rank matrix
$\mathbf{R}_n \in \mathbb{C}^{N_RL \times N_RL}$ at each iteration of the
MM--ADMM framework, which constitutes a major computational bottleneck.
To alleviate this burden, an approximate inverse
$({\mathbf{R}}_n^{(t)})^{-1}$ is computed online at iteration $t$ using the
Sherman--Morrison recursion~\cite{boydbook}
(Algorithm~\ref{alg:mm_admm}, line~\ref{Sherman_Morrison}). Let $\lambda_{n,d}$ and $\mathbf{s}_{n,d}$ denote the $d$th dominant
eigenvalue and its corresponding eigenvector of
$\boldsymbol{\Lambda}_n$, respectively. Specifically, the Sherman--Morrison recursion is initialized as $({\mathbf{R}}_{n,0}^{(t)})^{-1}
=\sigma_r^{-2}\mathbf{I}_{N_RL}$.

The $d$th rank-one update, for $d=1,\ldots,C$, is then given by
\begin{equation}
\label{eq:SM}
({\mathbf{R}}_{n,d}^{(t)})^{-1}
=
({\mathbf{R}}_{n,d-1}^{(t)})^{-1}
-
\frac{
({\mathbf{R}}_{n,d-1}^{(t)})^{-1}
\mathbf{f}_{d}^{(t)}
(\mathbf{f}_{d}^{(t)})^H
({\mathbf{R}}_{n,d-1}^{(t)})^{-1}
}{
1+
(\mathbf{f}_{d}^{(t)})^H
({\mathbf{R}}_{n,d-1}^{(t)})^{-1}
\mathbf{f}_{d}^{(t)}
},
\end{equation}
where $\mathbf{f}_d^{(t)} \triangleq \sqrt{\lambda_{n,d}}\,
\boldsymbol{W}_{n}^{(t)}\mathbf{s}_{n,d} \in \mathbb{C}^{N_RL \times 1}$
rotates the $d$th eigenvector $\mathbf{s}_{n,d}$ of $\boldsymbol{\Lambda}_n$ through 
the current transmit waveform matrix $\boldsymbol{W}_{n}^{(t)}$ weighted by $\sqrt{\lambda_{n,d}}$. 
Upon $C$ rank-one updates, we obtain 
${\mathbf{R}}_n^{(t)} = {\mathbf{R}}_{n,C}^{(t)}$.
It should be noted that this low-rank methodology drops the online matrix inversion complexity from
$\mathcal{O}((N_RL)^3)$ to
$\mathcal{O}(C(N_RL)^2)$ with $C \ll N_R L$.

\subsection{Proposed MM-ADMM Solution}

We develop a double-loop framework that pairs an outer MM loop
(Algorithm~\ref{alg:mm_admm}, lines~\ref{outerloop_t}--\ref{admm_initialization}) with an inner ADMM loop
(lines~\ref{admm_m}--\ref{admm_end}). 

The outer iteration index is denoted by $t$. 
In what follows, variables updated during the outer iteration carry the superscript $(\cdot)^{(t)}.$ Since the surrogate function on the r.h.s. of~\eqref{eq:surrogate} and the constraints 
in \eqref{eq:feasible} are separable with respect to the $L$ OFDM symbols, the MM step breaks into $L$ independent per-symbol subproblems. 
We define the per-symbol vector by stacking the subcarrier contributions (line~\ref{stack_waveform}) as
${\mathbf{q}}_\ell^{(t)} \triangleq 
[({\mathbf{q}}_{0,\ell}^{(t)})^T, \ldots, 
({\mathbf{q}}_{N-1,\ell}^{(t)})^T]^T
\in \mathbb{C}^{N_TN\times 1}$,
where ${\mathbf{q}}_{n,\ell}^{(t)} \in \mathbb{C}^{N_T\times 1}$ 
is the $\ell$th block of ${\mathbf{q}}_n^{(t)}$, and the 
per-symbol matrix as
${\mathbf{U}}^{(t)}_\ell \triangleq 
\operatorname{blkdiag}\{{\mathbf{U}}_{0,\ell}^{(t)},
\ldots,{\mathbf{U}}_{N-1,\ell}^{(t)}\}
\in \mathbb{C}^{N_TN\times N_TN}$,
where ${\mathbf{U}}_{n,\ell}^{(t)} \in \mathbb{C}^{N_T \times N_T}$
is the $\ell$-th symbol block of ${\mathbf{U}}_n^{(t)}$, i.e., ${\mathbf{U}}_{n,\ell}^{(t)} \triangleq 
\{{\mathbf{U}}_n^{(t)}\}_{[\ell N_T+1,\,(\ell+1)N_T],\,
[\ell N_T+1,\,(\ell+1)N_T]}$.
At iteration $t$, the per-symbol subproblem is then given by
\begin{equation}
\label{eq:mm_qp}
\boldsymbol{\omega}_\ell^{(t+1)}=\mathop{\mathrm{argmin}}_{\mathbf{\omega}_\ell\in\mathcal{F}_\ell}
\;
(\mathbf{\omega}_\ell)^H{\mathbf{U}}^{(t)}[\ell]\,\mathbf{\omega}_\ell
-
2\Re\!\left\{
({\mathbf{q}}_\ell^{(t)})^H\mathbf{\omega}_\ell
\right\},
\end{equation}
where $\mathcal{F}_\ell \!=\! \{\mathbf{\omega}_\ell \in \mathbb{C}^{N_TN \times 1} \mid 
\delta_{n,\ell,k} \geq \gamma_{k},
(\mathbf{\omega}_\ell)^H \mathbf{Q}_p \mathbf{\omega}_\ell \leq \eta_p,
|{x}_{r},\ell,q| \leq A_{\mathrm{papr}}\ ,\|\mathbf{\omega}_\ell\|_2^2 \!\leq\! P_T/L\}$ comprises the CI-QoS, angular-spectral coexistence, PAPR  
for symbol $\ell$  in~\eqref{eq:feasible} and relaxed power constraints. Note that $\mathcal{F}_\ell$ is the convex feasible set. 
The remaining constraint in  \eqref{eq:feasible} is the total power 
$\sum_{\ell=0}^{L-1}\|\mathbf{\omega}_\ell\|_2^2 = P_T$ constraint, which is non-convex.  
To handle this, we introduce a consensus variable $\mathbf{z}_\ell \in \mathbb{C}^{N_TN \times 1}$ to enforce the power constraint to $\mathbf{z}_\ell$ and the $t$th iteration of \eqref{eq:mm_qp} is solved by incorporating ADMM inner iterations (lines \ref{admm_m}-\ref{admm_end}).  
\begin{figure*}[t]
  \centering
  \includegraphics[width=\textwidth, height=0.17 \textheight, keepaspectratio]{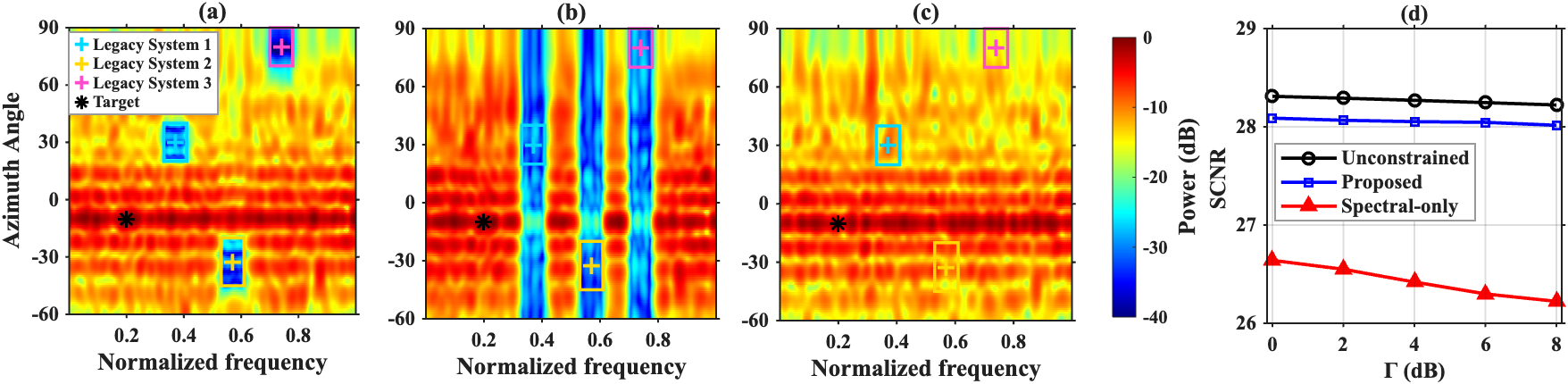}
  \caption{2D angular-spectral power maps for (a)~proposed
    (b)~spectral-only, (c)~unconstrained, and 
    (d)~Radar SCNR vs.\ communication QoS}
  \label{fig:combined}
\end{figure*}
The per-symbol ADMM inner iteration index is denoted by $m$ (line~\ref{admm_m}), and  variables updated during this iteration carry the superscript $(\cdot)^{<m>}$. 
Denoting augmented Lagrangian penalty parameter as $\rho > 0$ and scaled dual variable at the $m$th iteration as  $\boldsymbol{\lambda}^{<m>}_\ell \in \mathbb{C}^{N_TN \times 1}$, the per-symbol ADMM recursion (Algorithm~\ref{alg:mm_admm}, lines~\ref{admm_xupdate}-\ref{admm_dualupdate}) are formulated as
\vspace{-0.3cm}

\begin{subequations}
\small
\label{eq:admm_updates}
\begin{align}
\boldsymbol{\omega}_\ell^{<m+1>}
&=
\mathop{\mathrm{argmin}}_{\boldsymbol{\omega}_\ell\in\mathcal{F}_\ell}
\;
(\boldsymbol{\omega}_\ell)^H
\left(
{\mathbf{U}}_\ell^{(t)}
+
\frac{\rho}{2}\mathbf{I}
\right)
\boldsymbol{\omega}_\ell
\notag\\
&
~~ - \!
 2\Re 
\left\{ \!
\left(
{\mathbf{q}}_\ell^{(t)}
+
\frac{\rho}{2}
(\mathbf{z}^{<m>}[\ell]-\boldsymbol{\lambda}_\ell^{<m>})
\right)^H \!
\boldsymbol{\omega}_\ell \!
\right\},
\label{eq:x_update}
\\
\mathbf{z}_\ell^{<m+1>}
&=
\sqrt{\frac{P_T}{L}}
\frac{
\boldsymbol{\omega}_\ell^{<m+1>}
+
\boldsymbol{\lambda}_\ell^{<m>}
}{
\left\|
\boldsymbol{\omega}_\ell^{<m+1>}
+
\boldsymbol{\lambda}_\ell^{<m>}
\right\|_2
},
\label{eq:y_update}
\\
\boldsymbol{\lambda}_\ell^{<m+1>}
&=
\boldsymbol{\lambda}_\ell^{<m>}
+
\boldsymbol{\omega}_\ell^{<m+1>}
-
\mathbf{z}_\ell^{<m+1>}.
\label{eq:dual_update}
\end{align}
\end{subequations}

Because the $\mathbf{\omega}$-update in \eqref{eq:x_update} is a convex
quadratic program over $\mathcal{F}_\ell$, it is solved by CVX~\cite{boydbook} at each symbol $\ell$. 
The $\mathbf{z}$-update in \eqref{eq:y_update}
is the Euclidean projection of
$\mathbf{\omega}_\ell^{<m+1>}+\boldsymbol{\lambda}_\ell^{<m>}$
onto the per-symbol power sphere
$\|\mathbf{z}_\ell\|_2^2 = P_T/L$, enforcing the total
power constraint $\sum_{\ell=0}^{L-1}\|\mathbf{\omega}_\ell\|_2^2
= P_T$ across symbols. The dual update~\eqref{eq:dual_update} drives the consensus between $\mathbf{\omega}_\ell^{<m+1>}$ and
$\mathbf{z}_\ell^{<m+1>}$. The full procedure
is summarized in Algorithm~\ref{alg:mm_admm}.
\section{Simulation Results}

Numerical simulations are conducted to validate the proposed MM--ADMM
framework in Algorithm~\ref{alg:mm_admm} by evaluating radar SCNR,
localized angular--spectral coexistence, sensing--communication tradeoff
performance, and robustness under antenna sweep.
The BS is equipped with an equal number of transmit and receive antennas,
$N_T=N_R=12$, with half-wavelength antenna spacing, i.e., $\lambda/2$. A CPI comprises $N=32$ subcarriers with subcarrier spacing
$\delta_f=120$~kHz across $L=16$ OFDM symbols. The carrier frequency is
set to $f_c=24$~GHz and the total transmit power is $P_T=15$~W. The radar
receiver noise variance is normalized to $\sigma_r^2=1$ without loss of
generality. The BS simultaneously serves $K=3$ single-antenna downlink
communication users with communication noise power
$\sigma_c^2=-70$~dBm. The PAPR tolerance is set to $\varepsilon=1$. For
the radar sensing environment, a point target is placed at azimuth
$\theta_0=-10^\circ$, range $R_0=41.8$~m, and radial velocity
$v_0=-31.2$~m/s. The clutter environment is modeled using $N_c=60$
uniformly distributed isotropic clutter patches spanning the full azimuth
range $[-90^\circ,90^\circ]$, with Doppler velocities uniformly spanning
$[-10,10]$~m/s and delays covering the range $[0.5R_0,1.5R_0]$. A unit
clutter-to-noise ratio per patch is assumed. We use QPSK modulation
($\Omega=4$), i.e., $\Phi=\pi/4$.

Three legacy systems occupy the protected angle--frequency regions. The
frequency coordinates are expressed as normalized frequencies
$\bar{f}=f_{\mathrm{offset}}/(N\delta_f)\in[0,1]$, where
$f_{\mathrm{offset}}$ is the frequency offset from the first subcarrier
and $N\delta_f=3.84$~MHz is the total OFDM bandwidth. Legacy System~1
occupies $\theta\in[20^\circ,40^\circ]$ and
$\bar{f}\in[0.33,0.41]$; Legacy System~2 occupies
$\theta\in[-45^\circ,-20^\circ]$ and
$\bar{f}\in[0.53,0.61]$; and Legacy System~3 occupies
$\theta\in[70^\circ,90^\circ]$ and
$\bar{f}\in[0.70,0.78]$. The interference threshold $\eta_p$ is set to
enforce $-20$~dB suppression in each protected angular--spectral region
relative to the corresponding unconstrained waveform. Three ISAC methods are compared under identical initializations and
simulation parameters: (i)~\emph{Proposed}, the MM--ADMM
Algorithm~\ref{alg:mm_admm}; (ii)~\emph{Spectral-only}, which adapts the
spectral coexistence matrix $\mathbf{R}_I$ from~\cite{ADMMWave},
suppressing energy in each protected band across all angular directions;
and (iii)~\emph{Unconstrained}, an OFDM adaptation of the MM--ADMM
framework in~\cite{ISAC_2} with no protection imposed on the legacy
systems. We also include a \emph{Radar-only precision coexistence}
benchmark, in which the communication constraints are removed, resulting
in a MIMO radar waveform design subject to the angular--spectral
coexistence constraints.

Figs.~\ref{fig:combined}(a)--(c) present the two-dimensional (2D) angular-spectral
power maps by the three methods under a target communication QoS threshold of $\Gamma = 5$~dB (where $\gamma_k = \sqrt{\sigma_c^2\,\Gamma}\sin(\pi/\Omega)$).
In Fig.~\ref{fig:combined}(a), the proposed method synthesizes deep power attenuation notches that are precisely localized within the protected angle–frequency regions of the three legacy systems. 
The spectral-only benchmark in Fig.~\ref{fig:combined}(b) suppresses three
entire frequency bands across the full $180^\circ$ azimuth extent, sacrificing spatial DoFs in angular directions that require no protection. 
The unconstrained benchmark in Fig.~\ref{fig:combined}(c) distributes power freely with no protection to legacy regions.  

Before quantifying the SCNR--QoS tradeoff, we evaluate the 
feasibility of the three methods. For each method and over 
$M = 50$ independent channel realizations, we solve the 
max-threshold feasibility problem
$\Gamma^\star = \max~\Gamma~\text{subject to}~\boldsymbol{\omega} \in 
\mathcal{W}(\Gamma)$,
where $\mathcal{W}(\Gamma)$ denotes the constraint set comprising 
the CI, PAPR, transmit power, and angular-spectral leakage 
constraints at $-20$~dB suppression.
It is consistently observed that $\Gamma^\star$ exceeds $8$~dB 
for the proposed joint method across all $50$ channel trials, 
confirming that the QoS sweep range $\Gamma \in [0, 8]$~dB lies 
well within the feasible operating region. In contrast, the spectral-only baseline exhibits $\Gamma^\star<8$~dB
for some channel realizations. For these realizations, we increase
the transmit power beyond $15$~W.
\begin{figure}[t]
  \centering
  \includegraphics[width=0.55\columnwidth]{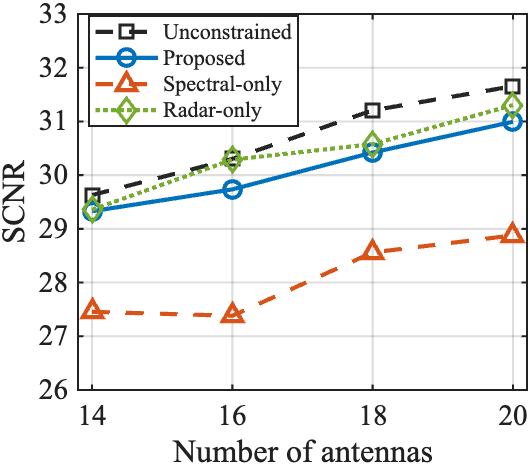}
  \caption{Radar SCNR versus the number of antennas $N_T=N_R$
  }
  \label{fig:antenna_sweep}
\end{figure}
Fig.~\ref{fig:combined}(d) presents the SCNR--QoS tradeoff for
$\Gamma\in\{0,2,4,6,8\}$~dB. All three methods exhibit decreasing
SCNR as $\Gamma$ increases, reflecting the fundamental ISAC tradeoff
between sensing and communication performance. The proposed method
achieves SCNRs of approximately $28.09$~dB and $28.02$~dB at $\Gamma=0$ and
$8$~dB, respectively, closely tracking the unconstrained benchmark,
which achieves approximately $28.31$~dB and $28.22$~dB at the same
operating points. The spectral-only baseline degrades more significantly
at high QoS, reaching only approximately $26.22$~dB at
$\Gamma=8$~dB, corresponding to about a $1.79$~dB loss relative to the
proposed method.  This performance gap occurs because high QoS requirements reduce the remaining optimization space for SCNR, making the spatial DoF 
preservation of our proposed joint constraints increasingly vital 
for maintaining robust radar sensing capability.

Fig.~\ref{fig:antenna_sweep} evaluates the impact of the number
of transmit and receive antennas by varying
$N_T=N_R\in\{14,16,18,20\}$. The proposed method improves from approximately
$29.33$~dB at $N_T=N_R=14$ to $30.99$~dB at
$N_T=N_R=20$, while the unconstrained benchmark increases
from approximately $29.63$~dB to $31.66$~dB. The spectral-only
method increases from approximately $27.46$~dB to $28.88$~dB,
remaining below the proposed method across all antenna
configurations. The radar-only benchmark achieves approximately
$29.36$~dB and $31.30$~dB at the same endpoint configurations,
indicating a modest sensing performance loss from imposing
the communication constraints in the proposed method.

\section{Conclusion and Future Works}
This paper proposed a dual-function waveform design framework for
MIMO-ISAC systems operating under angular-spectral coexistence
constraints with legacy systems. The resulting optimization problem was solved by developing a computationally efficient MM–ADMM algorithm. The proposed framework exploited spatial DoFs to restrict beamformed leakage strictly within the localized angle–frequency regions occupied by legacy systems. Simulation results demonstrated effective coexistence with negligible sensing and communication performance loss, and with increasing gains over spectral-only baselines under stringent communication QoS requirements. Several directions remain open for future work. Our immediate goal is to characterize how the design scales to larger space--frequency
dimensions, where the per-symbol convex subproblem becomes a major bottleneck. The study presented here also assumed perfect channel knowledge, and we plan to examine how the achievable suppression degrades under channel estimation error. Lastly, we intend to extend the formulation to cell-free massive MIMO, where legacy users can be
better protected but coexistence constraints must be enforced jointly across distributed access points.
\bibliographystyle{IEEEtran}
\bibliography{biblib}
\end{document}